\documentclass[11pt]{CGC2}
\usepackage[english]{babel}
\usepackage{verbatim}
\usepackage{psfrag}
\usepackage{graphicx}
\usepackage{braket}

\begin{document}
\Logo{BCRI--CGC--preprint, http://www.bcri.ucc.ie}

\begin{frontmatter}

\title{On the provable security of BEAR and LION schemes}

{\author{Lara Maines}} {\tt{(lara.maines@gmail.com)}}\\
 {{Department of Mathematics, University of Trento, Italy}}

{\author{Matteo Piva}}
{\tt{(matteo.piva@unitn.it) }}\\
{{Department of Mathematics,
 University of Trento, Italy }}

{\author{Anna Rimoldi}}
{\tt{(anna.rimoldi@gmail.com) }}\\
{eRISCS, Universite de la M\'{e}diterran\'{e}e, Marseille, France}

{\author{Massimiliano Sala}} {\tt{(maxsalacodes@gmail.com)}}\\
{ Department of Mathematics, University of Trento, Italy}

\runauthor{L.~Maines, M.~Piva, A.~Rimoldi, M.~Sala}

\begin{abstract}
BEAR, LION and LIONESS are block ciphers presented by Biham and
Anderson (1996), inspired by the famous Luby-Rackoff constructions of
block ciphers from other cryptographic primitives (1988). 
The ciphers proposed by Biham and Anderson are based on
one stream cipher and one hash function. Good properties of
the primitives ensure  good properties of the block cipher.
In particular, they are able to prove that their ciphers are immune to
any efficient known-plaintext key-recovery attack that can use as
input only one plaintext-ciphertext pair.
Our contribution is showing that these ciphers are actually immune to any
efficient known-plaintext key-recovery attack that can use as input
any number of plaintext-ciphertext pairs. We are able to get this
improvement by using slightly weaker hypotheses on the primitives.
We also discuss the attack by Morin (1996).
\end{abstract}

\end{frontmatter}

\section*{Introduction}
\label{intro}

In this paper we discuss three block ciphers, BEAR, LION and LIONESS, proposed in \cite{CGC-cry-art-AndersonBiham96} by Anderson and Biham, 
whose construction depends on one stream cipher and one hash function. 
These block ciphers are inspired by \cite{CGC-cry-art-LubyRackoff88} and 
present a three-round (for BEAR and LION) or four-round (for LIONESS) Feistel construction.   
In particular, we treat the provable security shown by them and 
provide some improvements.

In Section \ref{prel} we give some preliminaries, recalling in particular
BEAR's construction (Subsection \ref{prel_BEAR}) with the results by Anderson
and Biham, Th.~\ref{th_stream_anderson_biham} and Th.~\ref{th_hash_anderson_biham}, 
that ensure the non-existence of efficient attacks (of a very specific kind)
on BEAR if at least one of the two primitives is robust.
In this section we also recall LION's construction
(Subsection \ref{prel_LION}) and their
claimed results on LION,  Th.~\ref{th_stream_LION_anderson_biham}
and Th.~\ref{th_hash_LION_anderson_biham}, on the non-existence of similar
attacks.
We provide our proof for them, with slightly
weaker hypotheses.
 This preliminary section is concluded by a description
of LIONESS (Subsection~\ref{prel_LIONESS}).\\
In Section \ref{results} we give our results on BEAR, LION and LIONESS,
that show the non-existence of some more general attacks.
We also introduce two slight variations, BEAR2 and LION2, of BEAR
and LION, respectively.
We identify the hypotheses on the primitives that we need,
in particular highlighting the relation between key and hash function
in the keyed hash function.
In Subsection \ref{res_BEAR} we provide Th.~\ref{th_hash_bear}
 that improves 
Th.~\ref{th_hash_anderson_biham}, and Th.~\ref{th_BEAR_stream} on BEAR2, 
that improves Th.~\ref{th_stream_anderson_biham}.
In Subsection~\ref{res_LION} we provide Th.~\ref{th_lion_stream} that improves 
Th.~\ref{th_stream_LION_anderson_biham}, and Th.~\ref{th_lion2_stream} on LION2,
that extends and improves Th.~\ref{th_hash_LION_anderson_biham}.
Finally, in Subsection~\ref{res_LIONESS} 
we extend our results to LIONESS
in Th.~\ref{th_stream_lioness} and Th.~\ref{th_hash_lioness}.\\
In Section \ref{conc}, we discuss our results and draw our conclusions.
We also put in context the attack to BEAR and LION by Morin 
(\cite{CGC-cry-art-Morin96}).

\section{Preliminaries}
\label{prel}

 We use $\FF$ to denote $\mathbb{F}_2$ and typically when a capital $R$ or a capital $L$ appear, they mean elements of $\FF^r$ and $\FF^l$ respectively, with $r>l$. The encrypted/decrypted messages are of kind $(L_i,R_i)\in\FF^{l+r}$. The key space is denoted by $\mathcal{K}$.
Usually the key $K=(K_1,K_2)$ is composed of two subkeys, each of length greater than $l$, so $\mathcal{K}=\FF^k\times\FF^k$, $k\geq l$.

In this paper we  consider oracles able to recover the key using as input 
only a set of known plaintexts/ciphertexts.
We call ``single-pair'' any oracle so strong as to need only one pair
and ``multi-pair'' any other.

\subsection{Preliminaries on BEAR}
\label{prel_BEAR}
 The description of BEAR encryption/decryption is based on a keyed hash function $H_K$ and a stream cipher $S$ with the following properties.
\begin{enumerate}

 \item The keyed hash function $H_K(M)$
\begin{itemize}
 \item[(a)] is based on an unkeyed hash function $H'(M)$, in which we append and/or prepend the key to the message;
 \item[(b)]\label{1b} is one-way and collision-free, i.e. it is hard given $Y$ to find $X$ such that $H'(X)=Y$, and to find unequal $X$ and $Y$ such that $H'(X)=H'(Y)$;
 \item[(c)] is pseudo-random, in that even given $H'(X_i)$ for any set of inputs, it is hard to predict any bit of $H'(Y)$ for a new input $Y$. 
\end{itemize}
\item The stream cipher $S(M)$:
\begin{itemize}
 \item[($0$)] is pseudo-random 
(this condition is assumed but not listed in \cite{CGC-cry-art-AndersonBiham96});
 \item[(a)]\label{2a} resists key recovery attacks, in that it is hard to find the seed $X$ given $Y=S(X)$;
 \item[(b)] resists expansion attacks, in that it is hard to expand any partial stream of $Y$.
\end{itemize}
\end{enumerate}
\noindent
We note that conditions (1)-c and (2)-$0$ ensure respectively that $H_K$ and $S$ are pseudo-random, in order to obtain security against some distinguishing attacks in
a rather theoretical model 
(\cite{CGC-cry-art-LubyRackoff88} and \cite{CGC-cry-art-Lucks96}).

We recall the BEAR encryption/decryption scheme (here $k>l$).
\begin{center}
\begin{tabular}{|l|l|}
\hline
 ENCRYPTION 				& 	DECRYPTION		\\		
\hline					
$\overline{L}= L + H_{K_1}(R)$		&	$\overline{L}=L' + H_{K_2}(R')$ 	\\
$R'=R + S(\overline{L})$		&	$R=R' + S(\overline{L})$	\\
$L'= \overline{L} + H_{K_2}(R')$	&	$L = \overline{L} + H_{K_1}(R)$		\\
\hline
\end{tabular}
\end{center}
In \cite{CGC-cry-art-AndersonBiham96} Anderson and Biham claim 
the following results on one-pair oracles.
\begin{theorem}[Th. 1 of \cite{CGC-cry-art-AndersonBiham96}]\label{th_stream_anderson_biham}
An oracle which finds the key of BEAR, given one plaintext/ciphertext pair, can efficiently and with high probability find the seed $M$ of the stream cipher $S$ for any output $Y=S(M)$.
\end{theorem}

\begin{theorem}[Th. 2 of \cite{CGC-cry-art-AndersonBiham96}]\label{th_hash_anderson_biham}
An oracle which finds the key of BEAR, given one plaintext/ciphertext pair, can efficiently and with high probability find preimages and collisions of the hash function $H$.
\end{theorem}

\begin{remark}
\label{Hrules}
We observe that while proving Th. \ref{th_stream_anderson_biham} and Th. \ref{th_hash_anderson_biham} they {\bf only} need the following assumption on $H$:
\begin{center}
for most $R$'s the map $H^R:\FF^k\mapsto\FF^l$, $H^R(K)=H_K(R)$, is surjective.
\end{center}
This assumption is implied by the pseudo-randomness of the unkeyed
hash function $H'$ (and the bigger dimension of the key space), but
it is not equivalent to it. Indeed, it is easy to construct even
linear functions satisfying it.\\
On the other hand, {\bf no} hypothesis on the stream cipher $S$ is used.
\end{remark}

\begin{remark}\label{eff}
The word {\tt efficiently} in Th. \ref{th_stream_anderson_biham} and 
Th. \ref{th_hash_anderson_biham} might be confusing.
The oracle could need huge resources to work. A trivial example
is given by a brute force search of all keys. What Biham and Anderson
mean is that the attacker will need little computational effort
{\tt in addition to} any effort done by the oracle itself,
whatever large.
\end{remark}

As a direct consequence of Th.  \ref{th_stream_anderson_biham}
and Th. \ref{th_hash_anderson_biham} we have:
\begin{corollary}\label{cor_BEAR}
If it is impossible to find efficiently the seed of $S$ or 
it is impossible to find efficiently  preimages and
collisions of H, then no efficient (key-recovery) single-pair attack
exists for BEAR.
\end{corollary}

\subsection{Preliminaries on LION}
\label{prel_LION}

LION is quite similar to BEAR except that it uses the stream cipher twice and 
the hash function only once. For LION, $\mathcal{K}=\FF^{2l}$.\\
The requests are:
\begin{enumerate}
 \item The hash function $H(M)$:
 \begin{itemize}
  \item[(b)] is one-way and collision-free, i.e. it is hard given $Y$ to find $X$ such that $H'(X)=Y$, and to find unequal $X$ and $Y$ such that $H'(X)=H'(Y)$;
 \end{itemize}
 \item The stream cipher $S(M)$:
 \begin{itemize}
  \item[($0$)] is pseudo-random,
  \item[(a)] resists key recovery attacks, in that it is hard to find the seed $X$ given $Y=S(X)$;
  \item[(b)] resists expansion attacks, in that it is hard to expand any partial stream of $Y$.
 \end{itemize}
\end{enumerate}
\noindent
We note that Anderson and Biham here dropped 1-(a) and 1-(c).

We recall the LION encryption/decryption scheme (here $k=l$).
\begin{center}
\begin{tabular}{|l|l|}
\hline
 ENCRYPTION 				& 	DECRYPTION		\\
\hline											
$\overline{R}= R + S(L+K_1)$		&	$\overline{R}= R' + S(L'+K_2)$ \\
$L'=L+H(\overline{R})$			&	$L=L'+H(\overline{R})$	\\
$R'= \overline{R} + S(L'+K_2)$		&	$R= \overline{R} + S(L+K_1)$	\\
\hline
\end{tabular}
\end{center}

Results similar to Theorem \ref{th_stream_anderson_biham} and Theorem \ref{th_hash_anderson_biham} are claimed also for LION but without proof nor precise
statement. They write {\em the security reduction of LION proceeds similarly
to that of BEAR; an oracle which yields the key of LION will break both
its components.}.
Unfortunately, we have not been able to write down direct adaptions 
of the previous proofs, especially because here the property of $H$ 
as in Remark~\ref{Hrules} cannot be used and 
we do not see how one could get something similar
for the stream cipher.
Therefore, we now state precisely their claims, giving the weakest hypotheses 
we can exhibit.

\begin{theorem}
\label{th_stream_LION_anderson_biham}
Assume nothing on $H$ and $S$, except that they are set functions
$H:\FF^r\rightarrow\FF^l$ and $S:\FF^l\rightarrow\FF^r$.
An oracle $\mathcal{A}_1$ which finds the key of LION, given one plaintext/ciphertext pair, can efficiently and with high probability find the seed $M$ of the stream cipher $S$ for any particular output $Y=S(M)$.
\begin{proof}
 Let us choose a random input $(L,R)$. Let $K_1=M+L$. Then $S(L+K_1)=Y$.
  We can compute $\overline{R}=R+Y$, $L'=L+H(R+Y)$ and, by choosing any $K_2$, $R'=R+Y+S(L'+K_2)$. Then we give in input to the oracle the pair $\Set{(L,R),(L',R')}$ and $\mathcal{A}_1$ returns $(K_1,K_2)$, so we can immediately compute $M=L+K_1$. 
\end{proof}
\end{theorem}

To prove a similar theorem for the hash function, 
we need the following definition.
\begin{definition}
\label{def_good_pairing}
Let $H$ and $S$ be functions, $H:\FF^r\rightarrow\FF^l$ and $S:\FF^l\rightarrow\FF^r$, with $r\geq l$. We say that $(S,H)$ is a \emph{good pairing} if 
 for a random $Y\in \FF^l$ we have
 $H^{-1}(Y)\cap \mathrm{Im}(S) \neq\emptyset$.
\end{definition}
\noindent
We note that if at least one between $H$ or $S$ is pseudo-random, then
$(S,H)$ is a good pairing. However, we might have a good pairing 
even if none of the primitives is pseudo-random.
\\

We are ready for our interpretation of their claim on the link
between the security of LION and of the hash function.
\begin{theorem}\label{th_hash_LION_anderson_biham}
Assume that $(S,H)$ is a good pairing. An oracle $\mathcal{A}_1$ which finds the key of LION, given one plaintext/ciphertext pair, can efficiently and with high probability find preimages and collisions of the hash function $H$.
\begin{proof}
Since $r>l$ we can choose $\tilde{R}\notin\mathrm{Im}(S)$ with probability $\frac{2^r-2^l}{2^r}$ and calculate $H(\tilde{R})=\tilde{Y}\in\FF^{l}$. We can suppose $H^{-1}(\tilde{Y})\cap\mathrm{Im}(S)\neq \emptyset$ (else we can choose another $\tilde{R}$) and so there is an $X\in H^{-1}(\tilde{Y})\cap\mathrm{Im}(S)$. We consider as plaintext $(L,0)$ , where $L$ is any element of $\FF^l$ and $0\in\FF^r$. There exists $K_1$ such that $\overline{R}=S(L+K_1)=X$, because $X\in\mathrm{Im}(S)$. Thus $L'=L+H(\overline{R})=L+H(X)=L+\tilde{Y}$, because $X\in H^{-1}(\tilde{Y})$. It follows that for $K_2=L+\tilde{Y}+X$ we have $R'=X+S(L'+K_2)=X+X=0$. We give to $\mathcal{A}_1$ as input the pair $\Set{(L,0),(L+\tilde{Y},0)}$ and it returns $(K_1,K_2)$, so we can compute easily $X=S(L+K_1)$, finding a collision $H(\tilde{R})=H(X)=\tilde{Y}$. Note that $\tilde{R}\neq X$, since $\tilde{R}\notin\mathrm{Im}(S)$ and $X\in\mathrm{Im}(S)$.

To find a preimage, argue as above but with an arbitrary $Y\in\FF^l$.    
\end{proof}
\end{theorem}
The same considerations as in Remark \ref{eff} hold and a corollary
analogous to Cor. \ref{cor_BEAR} holds.

\subsection{Preliminaries on LIONESS}
\label{prel_LIONESS}

The third block cipher proposed in \cite{CGC-cry-art-AndersonBiham96} is LIONESS, 
which consists of four rounds and uses four independent keys, 
$K_1$, $K_3\in\FF^l$, $K_2$, $K_4\in\FF^k$, 
so $\mathcal{K}=\FF^l\times\FF^{k}\times\FF^{l}\times\FF^{k}$, for some $k$.

Anderson and Biham do not give explicit statements on LIONESS's security, but it is obvious from its construction
that any provable-security result for LION and/or BEAR directly extends to LIONESS, because any oracle
attacking LIONESS will be able to attack LION and BEAR, with possibly even less effort.

\pagebreak

\begin{center}
LIONESS \\
\begin{tabular}{|l|l|}
\hline
 ENCRYPTION 				 & 	DECRYPTION	\\
\hline
$\overline{R}=R+S(L+K_1)$		 & $\overline{L}=L'+H_{K_4}(R')$		\\
$\overline{L}= L + H_{K_2}(\overline{R})$ & $\overline{R}=R'+S(\overline{L}+K_{3})$\\
$R'=R + S(\overline{L}+K_3)$		 & $L = \overline{L} + H_{K_2}(\overline{R})$\\
$L'= \overline{L} + H_{K_4}(R')$		 &	$R=\overline{R}+S(L+K_1)$	\\
\hline							
\end{tabular}
\end{center}

For completeness, we can state the following obvious corollary.
\begin{corollary}
\label{th_LIONESS_anderson_biham}
Assume nothing on $H$ and $S$, except that they are set functions
$H:\FF^r\rightarrow\FF^l$ and $S:\FF^l\rightarrow\FF^r$ and that
for most $R$'s the map $H^R:\FF^k\mapsto\FF^l$, $H^R(K)=H_K(R)$, is surjective.
Then an oracle $\mathcal{A}_1$ which finds the key of LIONESS, given one plaintext/ciphertext pair, can efficiently and with high probability
both find the seed of $S$ and find preimages/collisions of $H$.
\end{corollary}

\section{Our improvements}
\label{results}
We propose a property for the keyed hash function, $H_K$, able to ensure 
the security from any key-recovery attack that uses plaintext/ciphertext pairs.
\begin{definition}\label{def_H_kresistant}
 Given a keyed hash function $\mathcal{H}=\Set{H_K}_{K \in \FF^k}$, $H_K:\FF^r\mapsto\FF^l$ for any $K\in\FF^k$, we say that $\mathcal{H}$ is \emph{key-resistant} if, given a pair $(Z,R)$ such that $Z=H_K(R)$ for a random $K$ and a random $R$, then it is hard to find $K$.
\end{definition}
Let us consider a keyed hash function $\mathcal{H}$ of kind $H_K=H'(f(K,R))$
for some injective function $f$. For practical purposes we want also that 
$(K,R)$ is easy to find from $f(K,R)$. 
In \cite{CGC-cry-art-AndersonBiham96} $f$ can be a concatenation,
but we do not need to be so restrictive.
Let $K$ and $R$ be random and consider the equation:
\begin{equation}
\label{Zeq}
Z=H'(f(K,R))=H_K(R)
\end{equation}
We note that \ref{1b}-(b) for $H'$ {\em implies} that (\ref{Zeq})
cannot be solved knowing $Z$. On the other hand,
the key-resistance of $\mathcal{H}$ {\em means} that (\ref{Zeq})
cannot be solved knowing $Z$ {\em and} $R$.
It may seem that there is a logical link between the two conditions,
but generally speaking there is none, as we are going to show:
\begin{itemize}
\item Suppose that  \ref{1b}-(b) does not hold. From $K$ and $R$ we get $Z$. 
Then we can solve $Z=H'(X)$. However, $Z$ can have many preimages 
($2^{r-l}$ on average), and so $X$ is likely to be outside $\mathrm{Im}(f)$.
The knowledge of $R$ cannot help here, except in discarding unwanted
preimages. Only if  \ref{1b}-(b) fails badly, that is, if we can get
efficiently all preimages of $Z$, then we will be able to solve 
(\ref{Zeq}) by discarding all preimages except that of the desired form
$(K,R)$.
\item Suppose that $\mathcal{H}$ is not key-resistant.
      If $Z$ does not come from $\mathrm{Im}(f)$ then there is
      no way the lack of key-resistance can help. But even if
      $Z=H(K,R)$ for some random $K$ and $R$, still the attacker does not
      know $R$ and so lack of key-resistance cannot help, unless
      the attacker is allowed to search in a brute-force effort
      the whole $\FF^r$, which is supposed to be hard in our context.
\end{itemize}

If $H'$ is pseudo-random, then $\mathcal{H}$ is clearly key-resistant.

As regards the stream-cipher, we can consider a similar notion in a slightly
more general situation, that is, when
$\mathcal{K}\subset\FF^{2l}$.
\begin{definition}\label{def_S_kresistant}
Let $\mathcal{K}\subset\FF^{2l}$.
 Given a stream cipher $S:\FF^l\mapsto\FF^r$, we say that 
$S$ is \emph{key-resistant} if, given a pair $(Z,L)$ such that $Z=S(L+K_1)$ 
for a random $(K_1,K_2)\in\mathcal{K}$ and a random $L\in\FF^l$, 
then it is hard to find $K_1$. 
\end{definition}
When $\mathcal{K}=\FF^{2l}$ we obviously have the equivalence between
\ref{2a}-(a) and the key-resistance, since translations act regularly.
\begin{remark}
We could change the definition of LION by having a different
action induced by the keys, that is, $S(\tau_K(L))$ instead
of $S(L+K)$, where $\{\tau_K\}_{K\in \mathcal{K}}\subset \mathrm{Sym}(\FF^l)$,
$\mathrm{Sym}(\FF^l)$ being the symmetric group acting on $\FF^l$.
 All subsequent results will still hold, provided the action
is regular.
\end{remark}

\subsection{Our improvements for BEAR} 
\label{res_BEAR}

It is possible to give an improvement of Theorem \ref{th_hash_anderson_biham}, 
passing from one-pair oracles to multi-pair oracles.

\begin{theorem}\label{th_hash_bear}
Let $n\geq 1$. Let $\mathcal{A}_{n}$ be an oracle able to find the key of BEAR given any set of $n$ plaintext-ciphertext pairs $\Set{\left((L_i,R_i),(L'_i,R'_i)\right)}_{1\leq i\leq n}$. Then $\mathcal{A}_n$ is able to solve efficiently any equation $Z=H_{K_{1}}(R)$, knowing $Z$ and $R$, for any random $R\in\FF^r$ and any random $K_1\in\FF^k$.   
\begin{proof}
 Let us choose a set $\Set{L_i}_{1\leq i\leq n}\subset \FF^l$ and consider the set of plaintexts $\Set{(L_i,R)}_{1\leq i\leq n}$. It is possible to generate a set of ciphertexts  $\Set{(L_i',R_i')}_{1\leq i\leq n}$ by choosing any sub-key $K_2$ and computing: $\overline{L}_i=L_i+Z$, $R'_i=R + S(L_i+Z)$, $L'_i= L_i+Z + H_{K_2}(R_i')$. With $\Set{\left((L_i,R),(L_i',R_i')\right)}_{1\leq i \leq n}$ as input, $\mathcal{A}_n$ outputs $K_2$, which was already known, and $K_1$, which was unknown.  
\end{proof}
\end{theorem}

Again, when we say {\em efficiently} we disregard any effort put by the oracle
itself (see Remark \ref{eff}).

\begin{corollary}
If the (keyed) hash function is key-resistant, no efficient
multi-pair oracle exists for BEAR.
\end{corollary}

Unfortunately we have not been able to obtain a direct improvement 
of Theorem~\ref{th_stream_anderson_biham}, but it is quite simple 
to modify BEAR in order to obtain a similar result also for the stream cipher.
Let us consider the following variation of BEAR's scheme, in which $\mathcal{K}\subset\FF^k\times\FF^l\times\FF^k$, for some $k$, with $K_1$, $K_3\in\FF^{k}$ and $K_2\in\FF^{l}$.
\begin{center}
\textbf{BEAR 2}
\end{center}
\begin{center}
\begin{tabular}{|l|l|}
\hline
 ENCRYPTION 				& 	DECRYPTION		\\
\hline					
$\overline{L}= L + H_{K_1}(R)$		&	$\overline{L}=L' + H_{K_3}(R')$ 	\\
$R'=R + S(\overline{L}+K_2)$		&	$R=R' + S(\overline{L}+K_2)$	\\
$L'= \overline{L} + H_{K_3}(R')$	&	$L = \overline{L} + H_{K_1}(R)$		\\
\hline
\end{tabular}
\end{center}
First we extend Th. \ref{th_hash_bear} from BEAR to BEAR2.
\begin{theorem}
\label{th_hash_bear2}
Let $n\geq 1$. Let $\mathcal{A}_{n}$ be an oracle able to find the key of BEAR 2 given any set of $n$ plaintext-ciphertext pairs $\Set{\left((L_i,R_i),(L'_i,R'_i)\right)}_{1\leq i\leq n}$. Then $\mathcal{A}_n$ is able to solve any equation $Z=H_{K_{1}}(R)$, knowing $Z$ and $R$, for any random $R\in\FF^r$ and any random $K_1\in\FF^k$.   
\begin{proof}
 Obvious adaption of the proof of Th. \ref{th_hash_bear}. We choose this time $K_2$ and $K_3$, we obtain $K_1$ again. 
\end{proof}
\end{theorem}
Now we are ready for the following result, linking the security of BEAR2 also to the properties of the stream cipher $S$, in a multi-pair context.
\begin{theorem}\label{th_BEAR_stream}
Let $n\geq 1$. Let $\mathcal{A}_{n}$ be an oracle able to find the key of BEAR2 given any set of $n$ plaintext-ciphertext pairs $\Set{\left((L_i,R_i),(L'_i,R'_i)\right)}_{1\leq i\leq n}$. Then $\mathcal{A}_n$ is able to solve any equation $Z=S(X+K_{2})$, knowing $Z$ and $X$, for any random $X\in \FF^l$ and any random $K_2\in\FF^l$. 
\begin{proof}
Let us choose a set $\Set{R_i}_{1\leq i\leq n}\subset \FF^r$ and two sub-keys $K_1$, $K_3$. It is possible to generate plaintext/ciphertext pairs by choosing $L_{i}= X + H_{K_1}(R_i)$ and computing: $\overline{L_i}= L_i + H_{K_1}(R_i)=X$, $R'_i= R_i+Z$, $L_i'=X + H_{K_3}(R_i')$. We give in input to $\mathcal{A}_n$ the set $\Set{(L_i,R_i),(L_i',R_i')}_{1\leq i\leq n}$, $\mathcal{A}_n$ returns $K_1$, $K_3$ which were already known, and $K_2$, which was unknown.  
\end{proof}
\end{theorem}

We can summarize our findings on BEAR2 in the following corollary.
\begin{corollary}
No efficient multi-pair key-recovery oracle exists for BEAR2
if  the hash function is key-resistant or the stream cipher is key-resistant.
\end{corollary}

\subsection{Our improvements for LION}
\label{res_LION}

A result similar to Theorem \ref{th_hash_bear} holds for LION.

\begin{theorem}\label{th_lion_stream}
Let $n\geq 1$. Let $\mathcal{A}_{n}$ be an oracle able to find the key of LION given any set of $n$ plaintext-ciphertext pairs $\Set{\left((L_i,R_i),(L'_i,R'_i)\right)}_{1\leq i\leq n}$. Then $\mathcal{A}_n$ is able to solve any equation $Z=S(L+K_{1})$, knowing $Z$ and $L$, for any random $L\in \FF^l$ and any random $K_1\in\FF^k$. 
\end{theorem}
\begin{proof}
Let us choose a set $\Set{R_i}_{1\leq i\leq n}\subset \FF^r$ and consider the set of plaintexts $\Set{(L,R_i)}_{1\leq i\leq n}$. It is possible to generate a set of ciphertexts  $\Set{(L_i',R_i')}_{1\leq i\leq n}$ by choosing any sub-key $K_2$ and computing: $\overline{R}_{i}=R_i+S(L+K_1)=R_i+Z$, $L'_i=L_i + H(R_i+Z)$, $R'_i= R_i+Z + S(L'_i+K_2)$. Using $\mathcal{A}_{n}$ we can find $K_2$, which was already known, and $K_1$, which was unknown.  
\end{proof}
As we have already seen for BEAR, we have not been able to extend a result 
similar to Theorem \ref{th_lion_stream} also for its hash functions, 
but it is quite simple to modify LION in order to obtain  it, 
as in the following table, where $K_1$, $K_3\in \FF^{l}$ and $K_2\in\FF^k$, 
and so $\mathcal{K}\subset\FF^l\times\FF^k\times\FF^l$ for some $k$.
\begin{center}
\textbf{LION2}\\
\begin{tabular}{|l|l|}
\hline
 ENCRYPTION 				& 	DECRYPTION			\\
\hline
$\overline{R}= R + S(L+K_1)$		&	$\overline{R}= R' + S(L'+K_3)$ 	\\
$L'=L+H_{K_2}(\overline{R})$		&	$L=L'+H_{K_2}(\overline{R})$	\\
$R'= \overline{R} + S(L'+K_3)$		&	$R= \overline{R} + S(L+K_1)$	\\
\hline
\end{tabular}
\end{center}
\begin{theorem}
\label{th_lion2_stream}
Let $n\geq 1$. Let $\mathcal{A}_{n}$ be an oracle able to find the key of LION2 given any set of $n$ plaintext-ciphertext pairs $\Set{\left((L_i,R_i),(L'_i,R'_i)\right)}_{1\leq i\leq n}$. Then $\mathcal{A}_n$ is able to solve any equation $Z=S(L+K_{1})$, knowing $Z$ and $L$, for any random $L\in \FF^l$ and any random $K_1\in\FF^k$. 
\end{theorem}
\begin{proof}
 Obvious adaption of the proof of Th. \ref{th_lion_stream}.
\end{proof}

\begin{theorem}\label{th_hash_lion}
Let $n\geq 1$. Let $\mathcal{A}_{n}$ be an oracle able to find the key of LION 2 given any set of $n$ plaintext-ciphertext pairs $\Set{\left((L_i,R_i),(L'_i,R'_i)\right)}_{1\leq i\leq n}$. Then $\mathcal{A}_n$ is able to solve any equation $Z=H_{K_{2}}(X)$, knowing $Z$ and $X$, for any random $X\in\FF^r$ and any random $K_2\in\FF^k$.   
\end{theorem}
\begin{proof}
 Let us choose a set $\Set{L_i}_{1\leq i\leq n}\subset \FF^l$ and any sub-keys $K_1$, $K_3\in\FF^l$. It is possible to generate plaintext/ciphertext pairs by choosing $R_i= X + S(L_i+K_1)$ and computing: $\overline{R}_i=R_i+S(L_i+K_1)=X+S(L_i+K_1)+S(L_i+K_1)=X$, $L'_i= L_i+Z$, $R'_i= X + S(L'_i+K_3)$ . We give in input to $\mathcal{A}_n$ the set $\Set{(L_i,R_i),(L_i',R_i')}$, $\mathcal{A}_n$ returns $K_1$, $K_3$, which were already known, and $K_2$, which was unknown.  
\end{proof}

We can summarize our findings on LION and LION2 in the following corollary.
\begin{corollary}
No efficient multi-pair key-recovery oracle exists for LION
if  the stream cipher is key-resistant.

No efficient multi-pair key-recovery oracle exists for LION2
if  the hash function is key-resistant or the stream cipher is key-resistant.
\end{corollary}

\subsection{Our improvements for LIONESS}
\label{res_LIONESS}

Since LIONESS combines the construction of LION and BEAR, it is quite obvious that any provable-security result holding for BEAR and LION still holds for LIONESS. For completeness, we give the formal proofs for our multi-pair results.
\begin{theorem}
\label{th_stream_lioness}
Let $n\geq 1$. Let $\mathcal{A}_{n}$ be an oracle able to find the key of LIONESS given any set of $n$ plaintext-ciphertext pairs $\Set{\left((L_i,R_i),(L'_i,R'_i)\right)}_{1\leq i\leq n}$. Then $\mathcal{A}_n$ is able to solve any equation $Z=S(L+K_{1})$, knowing $Z$ and $L$, for any random $L\in \FF^l$ and any random $K_1\in\FF^l$. 
\end{theorem}
\begin{proof}
Let us choose a set $\Set{R_i}_{1\leq i\leq n}\subset \FF^r$ and consider the set of plaintexts $\Set{(L,R_i)}_{1\leq i\leq n}$. It is possible to generate a set of ciphertexts  $\Set{(L_i',R_i')}_{1\leq i\leq n}$ by choosing any sub-keys $K_2, K_3, K_4$ and computing: $\overline{R}_{i}=R_{i}+Z$, $\overline{L}_{i}= L + H_{K_2}(\overline{R}_{i})$, $R'_i=R_i + S(\overline{L}_{i}+K_3)$ and $L'= \overline{L}_{i} + H_{K_4}(R'_{i})$. Using $\mathcal{A}_{n}$ we can find $K_2, K_3, K_4$, which were already known, and $K_1$, which was unknown.  
\end{proof}

\begin{theorem}\label{th_hash_lioness}
Let $n\geq 1$. Let $\mathcal{A}_{n}$ be an oracle able to find the key of LIONESS given any set of $n$ plaintext-ciphertext pairs $\Set{\left((L_i,R_i),(L'_i,R'_i)\right)}_{1\leq i\leq n}$. Then $\mathcal{A}_n$ is able to solve any equation $Z=H_{K_4}(R')$, knowing $Z$ and $R'$, for any random $R'\in\FF^r$ and any random $K\in\FF^k$.   
\end{theorem}
\begin{proof}
 Let us choose a set $\Set{L_i}_{1\leq i\leq n}\subset \FF^l$ and consider the set of ciphertexts $\Set{(L_i',R')}_{1\leq i\leq n}$. It is possible to generate a set of plaintexts  $\Set{(L_i,R_i)}_{1\leq i\leq n}$ by choosing any sub-keys $K_1,K_2,K_3$ and decrypting: $\overline{L_{i}}=L'_{i}+Z$, $\overline{R_{i}}=R' + S(\overline{L_{i}}+K_{3})$, $L_{i} = \overline{L_{i}} + H_{K_2}(\overline{R_{i}})$, $R_{i}=\overline{R_{i}}+S(L_{i}+K_1)$. Using $\mathcal{A}_n$ we can find $K_1$, $K_2$, $K_3$, which were already known, and $K_4$, which was unknown.  
\end{proof}

\section{Conclusions and further comments}
\label{conc}

Let us consider a keyed hash function with a very weak requirement,
i.e., that it is surjective both
fixing the key and with respect to the keys (see Remark \ref{Hrules}).
Anderson and Biham prove that no single-pair oracle exists for BEAR, 
under the assumption that ``the stream seed is difficult to recover OR the hash function is collision resistant OR the hash preimage is hard
to recover''.
We prove that no multi-pair oracle exists for BEAR
under the assumption that ''hash is key-resistant''.
We also suggest a slight modification of BEAR, BEAR 2, where we can prove
that no multi-pair oracle exists under the assumption that
``hash is key-resistant OR stream is key-resistant''.

The conclusions about key-recovery attacks for LION are quite similar
to those for BEAR. 
Anderson and Biham claim without proof that 
no single-pair oracle exists for LION under the assumption that
``the stream seed is difficult to recover OR the hash function is collision resistant OR the hash preimage is hard to recover''.
However, we have found no direct proof following their outline.
We prove that no single-pair oracle exists for LION under the assumptions that 
the stream seed is difficult to recover.
To prove the same thing with assumptions on the hash function, we need a condition that we call {\em good pairing}. 
Interestingly, this condition follows from the pseudo-random nature of $S$ OR the pseudo-random nature of $H$.
Given the good pairing for granted, we finish to prove their claim, that is,
no single-pair oracle exists for LION under the assumption that
``the hash function is collision resistant OR the hash preimage is hard
to recover''.
As in the case of BEAR, we prove that no multi-pair oracle exists 
for LION under the assumption that
``the stream cipher is key-resistant'', which is equivalent to 
``the stream preimage is hard to recover'' in many pratical situations.
We also suggest a slight modification of LION, LION 2, where we can
prove that no multi-pair oracles exist
under the assumption that ``the hash function is key-resistant OR the stream
cipher is key-resistant''.

As regards key-recovery attacks, LIONESS's virtues are 
the sum of LION's and BEAR's virtues.
So it is possible to prove the non-existence of one-pair oracles 
using the authors' assumptions, 
but we can indeed prove the non-existence of multi-pair oracles 
under only the key-resistance assumption.

We note that an attack by Morin (\cite{CGC-cry-art-Morin96}) has somehow
diminuished the confidence in the robustness of these schemes.
However, the attack succeeds only because its brute force search on the round
function contradicts the
key-resistance of the hash function and of the stream function.
So, whenever $\mathcal{H}$ or $S$ remain key-resistant, both
LION and BEAR are immune to such attacks.

\section*{Acknowledgements}
For their comments and suggestions the authors would like
to thank E.~Bellini, G.~Morgari and M.~Coppola.
The first three authors would like to thank their supervisor (the
fourth author).

This work has been supported by TELSY Elettronica e Telecomunicazioni, 
an Italian company working in Information and Communication
Security.

\bibliography{RefsCGC}

\providecommand{\bysame}{\leavevmode\hbox to3em{\hrulefill}\thinspace}
\providecommand{\MR}{\relax\ifhmode\unskip\space\fi MR }
\providecommand{\MRhref}[2]{%
  \href{http://www.ams.org/mathscinet-getitem?mr=#1}{#2}
}
\providecommand{\href}[2]{#2}
\begin{thebibliography}{Mor96}

\bibitem[AB96]{CGC-cry-art-AndersonBiham96}
R.~Anderson and E.~Biham, \emph{Two practical and provably secure block
  ciphers: {BEAR} and {LION}}, Proc. of {FSE}~1996, {LNCS}, vol. 1039, 
1996, pp.~113--120.

\bibitem[LR88]{CGC-cry-art-LubyRackoff88}
M.~Luby and C.~Rackoff, \emph{How to construct pseudorandom permutations from
  pseudorandom functions}, SIAM J. Comput. \textbf{17} (1988), no.~2, 373--386.

\bibitem[Luc96]{CGC-cry-art-Lucks96}
S.~Lucks, \emph{Faster {L}uby-{R}ackoff {C}iphers}, Proc. of {FSE}~1996, LNCS,
  vol. 1039, 1996, pp.~189--203.

\bibitem[Mor96]{CGC-cry-art-Morin96}
P.~Morin, \emph{Provably secure and efficient block ciphers}, Proc. of
  SAC~1996, 1996, pp.~30--37.

\end{thebibliography}

\end{document}